\documentclass[a4paper,USenglish]{lipics-v2021}
\nolinenumbers
\usepackage{algorithm}
\usepackage[noEnd=true]{algpseudocodex}

\usepackage{xspace}

\usepackage{amssymb}
\usepackage{pifont}

\newcolumntype{Y}{>{\raggedright\arraybackslash}X}

\makeatletter
\algdef{SE}[UPON]{Upon}{EndUpon}[1]{\algpx@startIndent\algpx@startCodeCommand\textbf{Upon}\ #1\ \algorithmicdo}{\algpx@endIndent\algpx@startCodeCommand\algorithmicend\ Upon}%

\ifbool{algpx@noEnd}{%
  \algtext*{EndUpon}%
  \apptocmd{\EndUpon}{\algpx@endIndent}{}{}%
}{}%

\pretocmd{\Upon}{\algpx@endCodeCommand}{}{}

\ifbool{algpx@noEnd}{%
  \pretocmd{\EndUpon}{\algpx@endCodeCommand[1]}{}{}%
}{%
  \pretocmd{\EndUpon}{\algpx@endCodeCommand[0]}{}{}%
}%
\makeatother

\newcommand{\invec}{\textbf{X}}

\def\Speer{\mu}

\def\ConstDecTree{\mbox{\sc Const\_Decision\_Tree}}
\def\Determine{\mbox{\sc Determine}}

\def\inline#1:{\par\vskip 7pt\noindent{\bf #1:}\hskip 10pt}
\def\dnsinline#1:{\par\vskip -7pt\noindent{\bf #1:}\hskip 10pt}

\def\inline#1:{\par\vskip 7pt\noindent{\bf #1:}\hskip 10pt}
\def\dnsinline#1:{\par\vskip -7pt\noindent{\bf #1:}\hskip 10pt}

\def\Source{source\xspace}
\def\Advers{\mathcal{A}}

\def\Inum{\mathcal{K}}

\def\indID{ind}
\def\indID{\ell}
\def\Exp{\hbox{\rm I\kern-2pt E}}

\long\def\commentstart #1\commentend{}

\def\byzfrac{\beta}
\def\goodfrac{\gamma}

\def\peers{peers\xspace}

\def\peer{peer\xspace}
\def\nonfaulty{nonfaulty\xspace}
\def\BYZ{\mathcal{F}} 

\def\res{res}

\def\DataQuery{\mbox{\tt Query}}
\def\Time{\mathcal{T}}
\def\Query{\mathcal{Q}}

\def\download{\mbox{\rm Download}}
\def\retrieve{{\textsf{Retrieve}}}

\def\phase{{\mathtt{phase}}}
\def\stage{{\mathtt{stage}}}

\def\metoo{``me neither''}
\def\EXEC{\mathtt{EX}}

\def\cH{\mathcal{H}}

\def\FS{\mbox{\sf FS}}

\def\Time{\mathcal{T}}

\def\Query{\mathcal{Q}}

\def\Message{\mathcal{M}}
 
\def\Protocol{\mathcal{P}}
\def\leftchild{\mbox{\sf\small left-child}}
\def\rightchild{\mbox{\sf\small right-child}}

\title{Distributed Download from an External Data Source in Asynchronous Faulty Settings}
\author{John Augustine}{Indian Institute of Technology Madras, India \and \url{https://cse.iitm.ac.in/~augustine/}}{augustine@iitm.ac.in}{https://orcid.org/0000-0003-0948-3961}{Supported by the Cybersecurity Centre, IIT Madras.}

\author{Soumyottam Chatterjee}{CISPA -- The Helmholtz Center for Information Security}{soumyottam.chatterjee@cispa.de}{https://orcid.org/0000-0002-0479-0690}{(Optional) author-specific funding acknowledgements}

\author{Valerie King}{University of Victoria, Canada \and \url{https://webhome.cs.uvic.ca/~val/}}{val@uvic.ca}{https://orcid.org/0000-0001-7311-7427}{}

\author{Manish Kumar}{Indian Institute of Technology Madras, India \and 
\url{https://sites.google.com/view/saumitr/home}}{manishsky27@gmail.com}{https://orcid.org/0000-0002-0414-7910}{}

\author{Shachar Meir}{Weizmann Institute of Science, Rehovot, Israel \and \url{https://shacharmeir007.github.io}
}{shachar.meir@weizmann.ac.il}{https://orcid.org/0009-0003-5007-047X}{}

\author{David Peleg}{Weizmann Institute of Science, Rehovot, Israel \and \url{https://www.weizmann.ac.il/math/peleg/} }{david.peleg@weizmann.ac.il}{https://orcid.org/0000-0003-1590-0506}{Venky Harinarayanan and Anand Rajaraman Visiting Chair Professor. The funds from this professorship enabled exchange visits between IIT Madras, India, and the Weizmann Institute of Science, Israel.}

\authorrunning{J. Augustine, S. Chatterjee, V. King, M. Kumar, S. Meir, D. Peleg, }
\Copyright{John E.\ Augustine, Soumyottam Chatterjee, Valerie King, Manish Kumar, Shachar Meir, and David Peleg} 



\EventEditors{}
\EventNoEds{2}
\EventLongTitle{}
\EventShortTitle{}
\EventAcronym{}
\EventYear{}
\EventDate{}
\EventLocation{}
\EventLogo{}
\SeriesVolume{xx}
\ArticleNo{1}

\begin{document}

\maketitle

\begin{abstract}
The distributed \emph{Data Retrieval (DR)} model consists of $k$ peers connected by a complete peer-to-peer communication network, and a trusted \emph{external data source} that stores an array $\invec$ of $n$ bits ($n \gg k$). Up to $\beta k$ of the peers might fail in any execution (for $\beta \in [0, 1)$). Peers can obtain the information either by inexpensive messages passed among themselves or through expensive queries to the source array $\invec$. In the DR model,  we focus on designing protocols that minimize the number of queries performed by any nonfaulty peer (a measure referred to as the \emph{query complexity}) while maximizing the resiliency parameter $\beta$.

The $\download$ problem requires each nonfaulty peer to correctly learn the entire array $\invec$. Earlier work on this problem focused on \emph{synchronous} communication networks and established several deterministic and randomized upper and lower bounds. Our work is the first to extend the study of distributed data retrieval to \emph{asynchronous} communication networks. We address the \download\ problem under both the Byzantine and crash failure models. We present query-optimal deterministic solutions in an asynchronous model that can tolerate any fixed fraction $\beta<1$ of crash faults.  In the \emph{Byzantine} failure model, it is known that deterministic protocols incur a query complexity of $\Omega(n)$ per peer, even under synchrony. We extend this lower bound to randomized protocols in the asynchronous model for $\beta \geq 1/2$, and further show that for $\beta < 1/2$, a randomized protocol exists with near-optimal query complexity. To the best of our knowledge, this is the first work to address the $\download$ problem in asynchronous communication networks.
\end{abstract}

\section{Introduction}

\subsection{Background and motivation}





The Data Retrieval Model (DR) was first introduced in \cite{ABMPRT24} to abstract the fundamental process of a group learning from a reliable external data source, where the data source is too large or it is too expensive to be learned individually (i.e., requires members of the group to collaborate), and some members of the group might crash during execution or act in other ways to deliberately sabotage the learning process. One key example of systems where this process takes place is \emph{blockchain oracles} \cite{ocr, DORA}. 
We address this Oracle data delivery process in detail later and present a method for improving its performance using the DR model and the protocols designed in this work. 

The DR model contains two entities: (i) a peer-to-peer network and (ii) an external data source in the form of an $n$ bit array \invec. There are $k$ peers, up to $\byzfrac$ fraction of which may be faulty (and at least $\goodfrac=1-\byzfrac$ fraction of which are \nonfaulty). Each peer has access to the content of the array through queries. The general class of \emph{retrieval problems} consists of problems $\retrieve(f)$ requiring every peer to output $f(\invec)$ for some computable function $f$ of the input array $\invec$. In this work, we focus on the most fundamental retrieval problem, $\retrieve(f_{id})$ where $f_{id}(\invec) =\invec$, referred to hereafter as the $\download$ problem\footnote{It is fundamental since every retrieval problem $\retrieve(f)$ can be solved by first performing download and then locally computing $f(\invec)$.}, where every peer needs to learn the entire input $\invec$.

In the absence of failures, the problem can be easily solved in a query-balanced manner. Even with failures, the problem can be trivially solved at the cost of a large number of queries, as the non-faulty peers can directly query all the bits. This solution is prohibitively expensive; thus, we focus on minimizing the number of queries made by each non-faulty peer.
For \emph{synchronous} systems with Byzantine faults, a lower bound of $\Omega(\byzfrac n)$  on query complexity for \emph{deterministic} $\download$ is shown in~\hbox{\cite{ABMPRT24}} for every $\byzfrac < 1$, followed by a matching upper bound when $\byzfrac<1/2$. This implies that in the presence of Byzantine faults, one cannot attain the ideal query complexity of $\frac{n}{\goodfrac k}$ without using randomization. 

In this work, we consider \download\ protocols in the \emph{asynchronous} setting for both the crash and Byzantine fault models. In the asynchronous Byzantine fault setting, we prove that, unlike the synchronous setting, where randomization can overcome the deterministic lower bound, $\Omega(n)$ queries per peer are required when $\byzfrac\geq 1/2$, even for randomized protocols.
We complement this lower bound with a protocol for the $\byzfrac<1/2$ regime that achieves a query complexity of $\tilde{O}\left(\frac{n}{ (\goodfrac - \byzfrac)k}\right)$, which, for a constant $\gamma-\beta$ , is within log factors of the generic lower bound of $\Omega(n/\goodfrac k)$. 

Turning our attention to the more benign setting of \emph{crash} faults (i.e., where all peers are honest but some $\beta$ fraction may stop functioning), the picture is brighter. For this model, it turns out that even in the asynchronous setting, one can get
efficient \emph{deterministic} \download\ protocols that achieve the optimal query complexity of $O\left(\frac{n}{\goodfrac k}\right)$, for any fraction $\byzfrac < 1$ of crashes.



 \subsection{The Model} 
\label{sec: model}

In the Data Retrieval (DR) model, the system consists of two components. The first is a collection of $k$ peers, each equipped with a unique ID from the range $[1,k]$, connected by a complete communication network (or clique).
The network provides \emph{peer-to-peer message passing}, namely, every peer can send at time $t$ a (possibly different) message of size at most $\phi$ bits to each other peer.

The second component of the DR model is an \emph{external data source}. The \emph{\Source} stores an $n$-bit input array $\invec=\{b_1,\ldots,b_n\}$.
%
%
It provides the peers with read-only access, allowing each peer to retrieve the data through queries of the form $\DataQuery(i)$, for $1\le i\le n$. 
The answer returned by the \Source\ would then be $b_i$, the $i^{th}$ element in the array.
This type of communication is referred to as \emph{\Source-to-peer} communication.

We consider asynchronous communication, where any communication (both among peer-to-peer and source-to-peer) can be delayed by any finite amount of time.
For randomized protocols, we use the following notion of \emph{cycles}.
\subparagraph*{Cycles.}
In the asynchronous model, there is no global notion of \emph{rounds}, as each peer operates at a different pace. Nevertheless, to describe our protocols and analyze their performance, it is convenient to
divide the \emph{local} execution of each peer $\Speer$ into (varying time) \emph{cycles}. Each such local cycle consists of the following stages.
\begin{itemize}
    \item Sending (0 or more) queries and getting answers.
    \item Sending (0 or more) messages.
    \item Waiting to receive messages.
\end{itemize}
We assume that local computation takes 0 time and can be performed at any point in a cycle. Moreover, when waiting for messages, after every message is received, the peer can adaptively decide whether to keep waiting for an additional message or continue to the next cycle. Note that the local cycle $r$ of peer $\Speer$ might coincide with a different local cycle $r'$ of another peer $\Speer'$.

In the absence of global time units, it is convenient to break the 
time axis into ``virtual blocks'' by defining $t^r$, for integer $r\ge 1$, 
as the first time any peer started its local cycle $r$.

Every message is of size at most $\phi$ bits, where $\phi$ is a system parameter.
Note that throughout the paper, we either set $\phi$ to a specific value, or leave it as a parameter, in which case increasing the message size parameter $\phi$ would result in faster protocols.

\subparagraph*{The adversary.}
Our analysis uses the notion of an \emph{adversary}, representing the adverse conditions in which the system operates, including the asynchronous communication and the possibility of failures.

The adversary has two types of operation. First, it can fail up to $\beta k$ peers, under the restriction that it can only fail a peer between its cycles (or before the first cycle), meaning that a peer can make random decisions in its current cycle without the adversary being able to react until the end of the cycle. Second, it can set the time $t_{\Speer, \Speer'}^r$ it takes a message sent by peer $\Speer$ in its local cycle $r$ to reach peer $\Speer'$, under the restriction that it must set the time $t_{\Speer, \Speer'}^r$ for every pair of peers $\Speer, \Speer'$, before time $t^r$. 
In other words, the adversary must set the latency of each message sent during a cycle $r$ before any peer starts cycle $r$. The adversary can also decide when every peer starts its execution (i.e., we do not assume a simultaneous start). 
Note that in the case of deterministic protocols, the notion of cycles is irrelevant, and we consider a standard adversary that can fail a peer at any point of the execution and can delay messages for any finite amount of time.

The adversary $\Advers$ selects the input data and determines the failure pattern of the peers. In the \emph{crash} failure model, the adversary's power is limited to crashing some of the peers in every execution of the protocol. Once a peer crashes, it stops its local execution of the protocol arbitrarily and permanently. This could happen \emph{in the middle of operation}, e.g., after the peer has already sent some, but perhaps not all, of the messages it was instructed by the protocol to send out at a given point in time. 
In contrast, in the \emph{Byzantine} failure model, a failed peer can deviate from the protocol in arbitrary ways.
We assume that the adversary can fail at most $\byzfrac k$ peers, for some given\footnote{We do not assume~$\byzfrac$ to be a fixed constant (unless mentioned otherwise).} $\byzfrac \in [0,1)$. 
We let $\goodfrac = 1-\byzfrac$, so there is (at least) a $\goodfrac$ fraction of \nonfaulty peers in every execution. Denote the set of faulty (respectively, \nonfaulty) peers in the execution by $\BYZ$. (resp., $\cH$). 
%
%

%
We assume that the adversary knows the protocol and hence can simulate it (up to random coins). 
We concentrate on the following complexity measures.
\begin{description}
\item[\emph{Query} Complexity ($\Query$):] 
the maximum number of bits queried by a \nonfaulty peer during the execution.
\item[\emph{Time} Complexity ($\Time$):]
the time it takes for the protocol to terminate.
\item[\emph{Message} Complexity ($\Message$):] the total number of messages sent by \nonfaulty peers during the execution.
\end{description}

We assume that queries to the \Source\ are 
the more expensive component in the system, 
so we focus mainly on optimizing the query complexity $\Query$. 
Measuring the \emph{maximum} cost per peer (rather than the \emph{total} cost) gives priority to a balanced load of queries over the \nonfaulty peers. 
Let us now formally define the \download\ problem. Consider a DR network with $k$ peers, where at most $\byzfrac k$ can be faulty, and a \Source that stores a bit array $\invec=[b_1, \dots, b_n]$. 
Each peer 
is required to learn $\invec$. Formally, each \nonfaulty peer $\Speer$ outputs a bit array $\res^\Speer$,
and it is required that, upon termination, $\res^\Speer[i]=b_i$ for every $i \in \{1, \cdots, n\}$ and $\Speer \in \cH$.

In the absence of failures, this problem can be solved by sharing the task of querying all $n$ bits evenly among the $k$ peers, yielding $\Query = \Theta(n/k)$. The message complexity is $\Message = \tilde{O}(n k)$, assuming small messages of size $\tilde{O}(1)$, 
and the time complexity is $\Time = \tilde{O}(n/k)$ since $\Omega(n/k)$ bits need to be sent along each communication link when the workload is shared.

\subsection{Related Work} 
\label{sec: related work}

The vast literature on fault-tolerant distributed computing includes extensive work on both crash faults and Byzantine faults. A foundational result by Fischer, Lynch, and Paterson (FLP)\cite{FLP} demonstrated that in asynchronous networks, even a single crash fault renders many fundamental problems—such as consensus and reliable broadcast—impossible to solve deterministically. Specifically, they showed that the adversary can indefinitely delay progress, violating the termination property. To circumvent this impossibility, many subsequent works in asynchronous settings have adopted randomized techniques\cite{Ben-or-async-agreement, KJ11-async-consensus} or relaxed the termination requirement~\cite{B87}.

Given the practical relevance of the asynchronous model, it has been widely adopted for studying fault-tolerant protocols under crash faults~\cite{fekete1987asynchronous, lynch2015consensus, coan1988compiler, barwell2025crash} and Byzantine faults~\cite{abraham2019asymptotically, ben1983another, bracha1987asynchronous, cachin2001secure, cachin2000random, canetti1993fast, duan2018beat, kapron2010fast, loss2018combining, tseng2014asynchronous}. Fundamental problems such as agreement and reliable broadcast have often served as building blocks in distributed protocol design, typically assuming that all input data is already locally available to the peers.

In contrast, our work addresses the \emph{Data Retrieval (DR)} model, where each peer must actively fetch data from a trusted external source and disseminate it to the rest of the peers while minimizing the cost associated with querying the source. We focus on the $\download$ problem, which requires every nonfaulty peer to correctly learn the entire data. Unlike classical problems, we show that $\download$—despite requiring termination—can be solved deterministically in asynchronous networks. This highlights that employing reliable broadcast or agreement as foundational components is not necessary to solve the $\download$ problem.
Moreover, this can be done with optimal query complexity for any fraction $\byzfrac < 1$ of crash-faults. In the 
more adversarial
Byzantine
setting,
we also design randomized protocols that solve the problem even when a majority of the peers are Byzantine.

To the best of our knowledge, this work is the first to study retrieval problems in the Data Retrieval (DR) model under asynchronous communication. The DR model has previously been explored in synchronous networks, most notably in~\cite{ACKKMP25a,ABMPRT24}. In particular,\cite{ABMPRT24} introduced the $\download$ problem, motivated by practical applications such as Distributed Oracle Networks (DONs), which form a crucial component of blockchain systems and employ protocols like OCR and DORA\cite{ocr, DORA}.

The primary focus of~\cite{ABMPRT24} was on minimizing query complexity in the presence of Byzantine faults in synchronous settings. They proved that any deterministic protocol must incur a query complexity of at least $\Query = \Omega(\byzfrac n)$ and matched this with an upper bound when $\byzfrac < 1/2$. On the randomized front, they proposed two protocols. The first tolerates any constant fraction $\byzfrac < 1$ of Byzantine faults but has suboptimal query complexity, achieving $\Query = O\left(\frac{n}{\goodfrac k} + \sqrt{n}\right)$ with high probability. The second protocol improves upon this by achieving near-optimal query complexity $\tilde{O}\left(\frac{n}{\goodfrac k}\right)$\footnote{We use the $\tilde{O}(\cdot)$ notation to hide $\byzfrac$ factors and polylogarithmic terms in $n$ and $k$.} with high probability, but it can only tolerate up to a $\byzfrac < 1/3$ fraction of Byzantine faults.

The companion paper~\cite{ACKKMP25a} builds upon the foundational work in~\cite{ABMPRT24} by closing several gaps in the randomized setting. It presents a randomized $\download$ protocol with query complexity $\Query = O\left(\frac{n}{\gamma k}\right)$, time complexity $\Time = O(n \log k)$, and message complexity $\Message = O(nk^2)$ for any crash-fault fraction $\beta \in [0,1)$. Additionally, it establishes a lower bound showing that in any \emph{single-round} randomized protocol, each peer must essentially query the entire input, indicating the inherent limitations of extremely fast protocols. 

Moving to two-round protocols,\cite{ACKKMP25a} proposes a randomized solution that achieves query complexity $\Query = O\left(\frac{n}{\gamma k} + \sqrt{n}\right)$ with high probability, improving upon the time complexity of \cite{ABMPRT24} under the same fault threshold $\beta < 1$. Notably, this protocol operates under a stronger adversarial model—referred to as \emph{Dynamic Byzantine}—in which the set of Byzantine peers may change from one round to another. Under this dynamic fault model, the authors further develop a protocol that achieves expected query complexity $\tilde{O}\left(\frac{n}{\gamma k}\right)$ (within logarithmic factors), at the expense of a higher time complexity of $O(\log k)$. In a more constrained variant of the model—called \emph{Dynamic Byzantine with Broadcast}—where even Byzantine peers must send the same message to all peers in a round, they show that the protocol achieves optimal (up to logarithmic factors) query complexity with high probability in the worst case, again with only polylogarithmic time complexity.


In contrast to prior work, we explore the $\download$ problem in asynchronous networks under both crash and Byzantine faults. In Section~\ref{sec:async_byz_faults}, we prove that in the presence of Byzantine faults, if $\byzfrac \ge 1/2$, then any asynchronous $\download$ protocol requires a query complexity of at least $\Query \geq n/2$, which increases to $\Query = n$ for deterministic protocols. For $\byzfrac < 1/2$, we adapt the randomized protocol from the synchronous model in~\cite{ACKKMP25a} to the asynchronous case, achieving near-optimal query complexity $\Query = \frac{n}{(\goodfrac - \byzfrac)k}$, which is efficient when $\goodfrac - \byzfrac$ is not too small. We also note that the simple deterministic protocol from~\cite{ABMPRT24}, with optimal query complexity $\Query = O(\byzfrac n)$, can be extended naturally to the asynchronous setting. In the crash-fault setting, we present a new deterministic protocol (Section~\ref{sec:async_crash_faults}) that achieves optimal query complexity $\Query = O\left(\frac{n}{\goodfrac k}\right)$ for any $\byzfrac < 1$, along with time complexity $\Time = O\left(n/\phi + \log_{1/\byzfrac} \phi\right)$ and message complexity $\Message = O(n \cdot k^2)$. A concise comparison between our results and prior synchronous protocols is provided in Table~\ref{tab:algo-comparison}.
 

\begin{table}[ht]
  \centering
  \small

  \begin{tabular*}{\textwidth}{l c c c c c}
    \hline
    \textbf{Synchrony} & \textbf{Query} & \textbf{Fault Model} & \textbf{Resilience} & \textbf{Type} & \textbf{Reference} \\
    \hline
    \multicolumn{6}{l}{\textbf{Prior Work}} \\
    \hline
    Synchronous & $\tilde{O}\left(\frac{n}{\gamma k}\right)$ & Byzantine & $\beta < \frac{1}{3}$ & Randomized & \cite{ABMPRT24} \\
    Synchronous & $\tilde{O}\left(\frac{n}{\gamma k} + \sqrt{n}\right)$ & Byzantine & $\beta < 1$ & Randomized & \cite{ABMPRT24} \\
    Synchronous & $\tilde{O}\left(\frac{n}{\gamma k}\right)$ & Byzantine & $\beta < 1$ & Randomized & \cite{ACKKMP25a} \\
    \hline
    \multicolumn{6}{l}{\textbf{This Paper}} \\
    \hline
    Asynchronous & $\Theta\left(\frac{n}{\gamma k}\right)$ & Crash & $\beta < 1$ & Deterministic & Thm~\ref{thm:async_multi_crash} \\
    Asynchronous & $\tilde{O}\left(\frac{n}{(\gamma - \beta)k}\right)$ & Byzantine & $\beta < 1/2$ & Randomized & Thm~\ref{thm:fast_expected_optimal} \\
    \hline
  \end{tabular*}
  \caption{
  An overview of our results with prior closely related work.
    }
    \label{tab:algo-comparison}
\end{table}

\subsection{Contributions}


%
%

We present the $\download$ problem in asynchronous communication networks, under both crash and Byzantine failures settings. In the crash-fault setting, our deterministic results are optimal w.r.t. to the resilience (for any $\beta<1$) and query complexity. Notice that this optimality also holds for randomized algorithms.
In the Byzantine failure setting, we provide deterministic and randomized lower bounds as well as upper bounds.
The main results are: 
\begin{enumerate}[(1)]
    \item {\bf Deterministic $\download$ in Crash-Fault:} We present a deterministic protocol for solving \download\ problem in the asynchronous setting with at most $f <k$ crash faults ($\gamma = 1-f/k$) with  $\Query = \Theta(\frac{n}{\goodfrac k})$, $\Time = O\left(\frac{n}{\phi} + \log_{k/f}(\phi)\right)$ and $\Message = O(nk^2)$ where $\phi$ is the message size. Our result achieves the optimal query complexity for any fraction of crash fault, $\beta<1$.
    
    \item {\bf Deterministic Lower Bound in Byzantine Fault:} We show that for $\byzfrac\geq 1/2$, every deterministic asynchronous $\download$ protocol that is resilient to Byzantine faults requires $\Query = n$.

    \item {\bf Deterministic $\download$ in Byzantine Fault:} We show that for $\byzfrac<1/2$, there exists a deterministic  asynchronous protocol that solves \download\ with $Q=O(\byzfrac n)$, $\Time=O\left(\frac{\byzfrac n}{\phi}\right)$ and $\Message = O(f\cdot  n)$.

    \item {\bf Randomized Lower Bound in Byzantine Fault:} We show that for any randomized asynchronous \download\ protocol where $\byzfrac\geq 1/2$, there does not exist any execution in which every peer queries less than or equal $n/2$ bits.

    \item {\bf 2-cycle Randomized $\download$ in Byzantine Fault:} We present a 2-cycle asynchronous randomized protocol for \download\ with $\Query=O\left( \sqrt{\frac{n}{\gamma-\beta}}+ \frac{n \log n}{(\gamma-\beta) k} \right)$ and $\Message = O(k^2)$ where $\byzfrac\leq 1/2$, and the message size is $\phi = O(\frac{n}{\goodfrac k})$.

    \item {\bf Randomized $\download$ in Byzantine Fault:} We present a $O\left(\log \left(\frac{\goodfrac k}{\ln n}\right)\right)$-cycle protocol that computes \download\ whp in the point-to-point model having expected query complexity $\Query=O\left(\frac{n \log n}{(\gamma - \byzfrac) k}\right)$ and $\Message = O\left(\log \left(\frac{\goodfrac k}{\ln n}\right)k^2\right)$ where $\byzfrac\leq 1/2$, and the message size is $\phi = O(n)$. 
\end{enumerate}
\section{Deterministic \download\ in the Asynchronous Model with Crash Faults} \label{sec:async_crash_faults}


In this section, we present deterministic protocols that solve the \download\ problem in an asynchronous setting. First, in Section \ref{ss: 1 crash}, we show a deterministic protocol that handles a single crash failure. Then in Section \ref{ss: f crashes}, we show how to extend the solution to $f$ crashes for $f>1$.

\subsection{\download\ with At Most One Crash}
\label{ss: 1 crash}

We start 
with a protocol that solves \download\ with at most one crash fault. This serves as an introduction to our main protocol, handling an arbitrary number of crashes.

The protocol runs in two phases. In each phase, every peer maintains a list of assigned indices. Each phase has three stages, and every message contains the local phase number and stage number. We describe a single phase and particularly the operation of peer $\Speer_i$.

In stage 1, the peer $\Speer_i$ queries all assigned bits that are still unknown and sends stage-1 messages containing the assigned bits' values to all other peers. (Assuming it is necessary to send the assigned bits in multiple packets, each packet would also include the packet number.) 

In stage 2, every peer $\Speer_i$ waits until it receives all stage-1 messages (according to its local assignment) from at least $k-1$ peers (waiting for the last peer risks deadlock, in case that peer has crashed). When that condition is met, $\Speer_i$ sends a stage-2 message containing the index $j$ of the \emph{missing} peer, namely, the peer $\Speer_j$ from which it didn't receive all stage-1 messages during the current phase. When $\Speer_i$ receives a stage-2 message containing the index $j$, it sends a stage-2 response containing either the bits assigned to $\Speer_j$ (if it heard from $\Speer_j$) or \metoo\ if it didn't hear from $\Speer_j$ during stage 2. (In case $\Speer_i$ hasn't finished waiting for stage-1 messages, it delays its response until it is finished.)

Finally, in stage 3, every peer $\Speer_i$ waits until it collects at least $k-1$ stage-2 responses. When that happens, there are two cases. Either $\Speer_i$ has received only \metoo\ messages, in which case it reassigns the bits of $\Speer_j$ evenly to all peers and starts the next phase, or $\Speer_i$ has received $\Speer_j$'s bits from at least one peer, in which case it goes into \textit{completion} mode, which means that in the next phase, $\Speer_i$ acts as follows. In stage-1, it sends all the bits in stage-1 messages. In stage-2, it doesn't send stage-2 messages, and in stage-3, it doesn't wait for stage-2 responses.

After two phases, every peer terminates.

Note that if a peer $\Speer$ receives a message from 
(some $\Speer'$ which is in) a later phase, it stores it for future use, and if it receives a message from an earlier phase, it stores
the bits the message (possibly) holds and evaluates whether or not to enter completion mode. 
See Algorithm~\ref{alg:single_crash_async_down} for the formal code.

\begin{algorithm}
    \caption{Asynchronous \download\ one crash (code for peer $\Speer_i$)} \label{alg:single_crash_async_down}
    \begin{algorithmic}[1]
    \State $Mode \gets$ Active
    \For{$t \in \{1, 2\}$}
        \Procedure{Stage 1}{}
        \State Query all unknown assigned bits.
        \If{$Mode =$ Active}
            \State Send a stage-1 message with the assigned bits to every other peer. 
        \Else
            \State Send a stage-1 message with all known bits.

        \EndIf

        \State (In the CONGEST model, this message may need to be broken into packets of size $O(\log n)$.)
    \EndProcedure
    \Procedure{Stage 2}{}
        \State Maintain the following set:
        \Statex $H = \{j \mid \text{received all stage-1 messages from } \Speer_j \}$
        \State When $|H| \geq k-1$, if $Mode=$Active, send $J_f(i)$ s.t. $\{J_f(i)\} = \{1, \dots, k\} \setminus H$ to all other peers. 
        \State  \textbf{When}
        receiving a stage-2 message $J_f(i')$ from peer $\Speer_{i'}$:
        \Statex Send back a stage-2 response containing \metoo\ if $J_f(i') = J_f(i)$ and $\Speer_{J_f(i')}$'s assigned bits otherwise. 
    \EndProcedure
    \Procedure{Stage 3}{}
        \If{$Mode=$Active}
        \State Collect stage-2 responses until receiving at least $k-1$ messages.
        \If{Received only \metoo\ responses}
            \State Reassign $\Speer_{J_f(i)}$'s bits to peers $I =\{1, \dots, k \} \setminus \{J_f(i)\}$, and start the next phase. 
        \Else
            \State $Mode \gets$ completion
        \EndIf
        \EndIf
    \EndProcedure
    \Procedure{messages from different phases}{}
        \If{
        Received a message from a different phase}
            \State If it is 
            from a later phase, store it for future use.
            \Statex If it is from a previous phase, evaluate its content and update your 
            known bits accordingly. If you have no more unknown bits, go into completion mode. 
        \EndIf
    \EndProcedure
    \EndFor
    \end{algorithmic}
\end{algorithm}

We use the following facts to show the correctness and complexity of the protocol.
\begin{observation} \label{obs:overlap}
    (Overlap Lemma) Assuming $2f<k$, every two sets of $k-f$ peers must overlap at least one peer.
\end{observation}
\begin{proof}
Exactly $f$ distinct peers are not present in one of the sets, hence at least $k-2f \geq 1$ peers in the other set must not be distinct, thereby being in both sets.
\end{proof}

The following observation holds after each of the two phases, although a stronger property holds after phase 2, namely, each \nonfaulty peer knows all input bits.

\begin{observation} \label{obs:peers_missing}
After Stage 2, every peer lacks bits from at most one ``missing'' peer.
\end{observation}

This observation is obvious from the protocol since each peer receives all stage-1 messages from at least $k-1$ peers.

\begin{lemma} \label{lem:stage_3_uniquness}
After stage 3, if two different peers $\Speer_i$ and $\Speer_{i'}$ lack bits from missing peers $\Speer_j$ and $\Speer_{j'}$ respectively, then $j=j'$.
\end{lemma}

\begin{proof}
Let $\Speer_i$ and $\Speer_{i'}$ be peers as required by the premises of the lemma. $\Speer_i$ received $k-1$ \metoo\ messages for $j$ and $\Speer_{i'}$ received $k-1$ \metoo\ messages for $j'$. By the Overlap Lemma \ref{obs:overlap}, there is at least one peer that sent \metoo\ to both $\Speer_i$ and $\Speer_{i'}$ for $j$ and $j'$ respectively, yet, by Observation \ref{obs:peers_missing}, each peer lacks bits from at most one missing peer. Hence, $j=j'$.
\end{proof}

\begin{theorem} \label{thm:async_single_crash}

In the asynchronous model with one crash fault, there is a deterministic protocol that solves \download\ with $\Query =\frac{n}{k} + \lceil\frac{n}{k(k-1)}\rceil$,
$\Time = O\left(\frac{n}{k\phi}\right)$ and
$\Message = O(nk/\phi)$.
\end{theorem}
\begin{proof}
By Observation \ref{obs:peers_missing}, after stage 2 of phase 1, each peer $\Speer_i$ lacks bits from at most one missing peer. By Lemma \ref{lem:stage_3_uniquness}, after stage 3 of phase 1, every peer $\Speer_i$ that still lacks some bits has the same missing peer $\Speer_j$. Consider such a peer $\Speer_i$. At the end of stage 3 of phase 1, $\Speer_i$ reassigns the bits of $\Speer_j$ evenly to $I =\{1, \dots, k\} \setminus \{j\}$. Every peer in $I$ either lacks some bits from $\Speer_j$ or is in completion mode, so in the following phase 2, they will send stage-1 messages consistent with the local assignment of $\Speer_i$. Hence, in stage 2 of phase 2, either $\Speer_i$ receives a phase-2 stage-1 message from $\Speer_j$, meaning it also receives a phase-1 stage-1 message from $\Speer_j$, or it receives all phase-2 stage-1 messages from all peers in $I$. In both cases, there are no unknown bits after stage 2 of phase 2.

In terms of query complexity, each peer queries $n/k$ times in phase 1. Because each peer lacks bits from at most one peer, it has at most $n/k$ unknown bits. At the end of stage 3 of phase 1, every such peer reassigns those $n/k$ bits evenly to the remaining $k-1$ peers, resulting in additional $n/(k(k-1))$ queries per peer. In total, the query complexity is $n/k + n/(k(k-1))$.
\end{proof}

\subsection{Extending the Result to \texorpdfstring{$f$}{} Crashes}
\label{ss: f crashes}

In this subsection, we present a protocol that extends Algorithm~\ref{alg:single_crash_async_down} to a protocol that can tolerate up to $f$ crashes for any $f<k$.
The main difficulty in achieving tolerance with up to $f$ crashes is that in the presence of asynchrony, one cannot distinguish between a slow peer and a crashed peer, making it difficult to coordinate.
    
As in Algorithm~\ref{alg:single_crash_async_down}, Algorithm~\ref{alg:f_crash_async_down} executes in phases, each consisting of three stages. Each peer $\Speer$ stores the following local variables. (We omit the superscript $\Speer$ when it is clear from the context.)
\begin{itemize}
    \item $phase(\Speer)$: $\Speer$'s current phase.
    \item $stage(\Speer)$: $\Speer$'s current stage within the phase.
    \item $H_p^\Speer$: the \emph{correct set} of $\Speer$ for phase $p$, i.e., the set of peers $\Speer$ heard from during phase $p$.
    \item $\sigma_p^\Speer$: the \emph{assignment function} of $\Speer$ for phase $p$, which assigns the responsibility for querying each bit $i$ to some peer $\Speer'$.
    \item $\res^\Speer$: the output array.
\end{itemize}

In the first stage of phase $p$, each peer $\Speer$ queries bits according to its local assignment $\sigma_p$ and sends a $\phase \ p \ \stage \ 1$ request (asking for bit values according to $\sigma_p$, namely $\{i \mid \sigma_p(i) = \Speer'\}$) to every other peer $\Speer'$ and then continues to stage 2. Upon receiving a $\phase \ p \ \stage \ 1$ request, $\Speer$ waits until it is at least in stage 2 of phase $p$ and returns the requested bit values that it knows.

In stage 2 of phase $p$, $\Speer$ waits until it hears from at least $|H_p^\Speer| \geq k-f$ peers (again, waiting for the remaining $f$ peers risks deadlock). Then, it sends a $\phase \ p \ \stage \ 2$ request containing the set of peers' IDs $F_p^\Speer = \{1, \dots, k\}\setminus H_p^\Speer$ (namely, all the peers it didn't hear from during phase $p$) and continues to stage 3. Upon receiving a $\phase \ p \ \stage \ 2$ request, $\Speer$ waits until it is at least in stage 3 of phase $p$, and replies to every peer $\Speer'$ as follows. For every $j\in F_p^{\Speer'}$, it sends $\Speer_j$'s bits if $j\in H_p^{\Speer'}$ and \metoo\ otherwise.

In stage 3 of phase $p$, $\Speer$ waits for $k-f$ $\phase \ p \ \stage \ 2$ responses. Then, for every $j\in F_p^\Speer$, if it received only \metoo\ messages, it reassigns $\Speer_j$'s bits evenly between peers $1, \dots, k$. Otherwise, it updates $\res$ in the appropriate indices. Finally, it continues to stage 1 of phase $p+1$. Upon receiving a $\phase \ p \ \stage \ i$ response, $\Speer$ updates $\res$ in the appropriate index and updates $H_p$ for every bit value in the message. See Algorithm~\ref{alg:f_crash_async_down} for the formal code.

Before diving into the analysis, we overview the following intuitive flow of the protocol's execution. At the beginning of phase 1, the assignment function $\sigma_1$ is the same for every peer. Every peer is assigned $n/k$ bits, which it queries and sends to every other peer. Every peer $\Speer$ hears from at least $k-f$ peers, meaning that it has at most $f\cdot n/k$ unknown bits after phase 1. In the following phases, every peer $\Speer$ reassigns its unknown bits uniformly among all the peers, such that the bits assigned to every peer $\Speer'$ are either known to it from a previous phase or $\Speer'$ is about to query them in the current phase (i.e., $\Speer'$ assigned itself the same bits). Hence, after every phase, the number of unknown bits diminishes by a factor of $f/k$. After sufficiently many phases, the number of unknown bits will be small enough to be directly queried by every peer.

\begin{algorithm}
    \caption{Async \download\ version 2 for peer $\Speer$} \label{alg:f_crash_async_down}
    \begin{algorithmic}[1]




    
    \State \textbf{Local variables} 

    \State $phase(\Speer)$, initially 0 \Comment{This is the present phase of $\Speer$}

    \State $stage(\Speer)$, initially 1 \Comment{This is the present stage of $\Speer$}

    \State $H_p$, $p \in \mathbb{N}$, initially $\emptyset$      \Comment{The set of peers $\Speer$ heard from during phase $p$}
    \State $\sigma_p(i)$, $p \in \mathbb{N}$, initially $\sigma_p(i) \gets \Speer_{1+\lceil i / \frac{n}{k}\rceil}$  \Comment{The assignment function of $\Speer$ in phase $p$}
    \State $\res[i]$, $i\in \{1,\dots, n\}$, initially $\bot$ \Comment{Output array}

    \State 

    \Upon{entering stage 1 of phase $p$}
        \State \Comment{--------- stage 1 (start) ---------------------------------}
        \State Query all unknown assigned bits, $\{i \mid \res[i]=\bot, \sigma_p(i) = \Speer\}$.
        \State Send a $\phase \ p \ \stage \ 1$ request containing $\{i \mid \res[i] = \bot, \sigma_p(i) = \Speer'\}$ to every peer
        \State Set $stage(\Speer)\gets 2$

        \State \Comment{--------- stage 2 (start) ---------------------------------}
        \State Wait until $|H_p| \geq k-f$ \label{line:wait_2} 
        \State  Send a $\phase\ p\ \stage\ 2$ request containing $F_p=\{1,\dots, k\}\setminus H_p$ to every peer $\Speer'$ 
        \State Set $stage(\Speer)\gets 3$
        \State \Comment{--------- stage 3 (start) ---------------------------------}
        \State Wait until received at least $k-f$ $\phase$ $p$ $\stage$ $2$ responses.\label{line:wait_3} 
        \For{$j \in F_p$} 
            \If{Received only \metoo\ responses for $j$}
                \State Let $i_0, \dots, i_{n'-1}$ be the indices such that $\sigma(i_l) = \Speer_j$, $0\leq l\leq n'-1 $.
                \State Set $\sigma_{p+1}(i_l) \gets \Speer_{1+\lceil l / \frac{n'}{k} \rceil}$, $0\leq l\leq n'-1 $.\label{line:reassign}
                \Comment{ Reassign $\Speer_j$'s bits to all $\{1, \dots, k \}$} 
            \Else
                \For{$i \in \{i \mid \sigma(i) = \Speer_j\}$}
                    \State $\res[i] \gets b_i$
                \EndFor
                \Comment{Update $\res$ in the appropriate indices}
            \EndIf
        \EndFor
        \State Set $phase(\Speer)\gets phase(\Speer)+1$, $stage(\Speer)\gets 1$ 
        
    \EndUpon
    
    \State
    
    \Upon{seeing a $\phase \ p \ \stage \ 1$ request for bit Set $B$}
        \State Store the request until $phase(\Speer)=p$ and $stage(\Speer)\geq 2$ or $phase(\Speer) >p$ 

        \State Send back a $\phase \ p \ \stage \ 1$ response containing $\{\langle i, \res[i] \rangle \mid i \in B\}$
        
    \EndUpon

    \State

    \Upon{seeing a $\phase \ p \ \stage \ 2$ request containing $F$ from $\Speer'$}{}
        \State  Store the request until $phase(\Speer)=p$ and $stage(\Speer)\geq 3$ or $phase(\Speer) >p$ 
        \For{$j \in F_p$}
            \State  Send back a $\phase\ p \ \stage \ 2$ response containing \metoo\ if $j\notin H_p$ and $\Speer_j$'s assigned bits otherwise.
        \EndFor  
    \EndUpon

    \State
    \Upon{receiving a $\phase \ p$ $\stage \ i$ response}{}
        \State For every bit value in the message $\langle i, b_i \rangle$, set $\res[i] \gets b_i$
        \State Update $H_p \gets \{j \mid \res[i]\neq \bot \ \ \forall i :\sigma_p(i) = j\}$
    \EndUpon
    \State 
    



    \Upon{$phase(\Speer) = \log_{k/f}(n/k)$ or $H_p=\{1, \dots, k\}$}{}
    \label{line:termination_condition}
    \State Query all unknown bits
        \State Send $\res$ to every other peer and Terminate
    \EndUpon
    
\end{algorithmic}
\end{algorithm}

We start the analysis by showing some properties of the relations between local variables.
\begin{observation} \label{obs:no_missing_mahines_no_missing_bits}
    For every \nonfaulty peer $\Speer$, if $H_p^\Speer=\{1, \dots, k\}$ for some phase $p\geq 0$ then $\res^\Speer=\invec$
\end{observation}
\begin{proof}
    Let $p\geq 0$ be such that $H_p=\{1, \dots, k\}$,
    and consider $1\leq i\leq n$. There exists some $1\leq j\leq k$ such that $\sigma_p(i) = \Speer_j$. Since $j \in H_p^\Speer$, $\Speer$ has heard from $\Speer_j$, so $\res^\Speer[i] \neq \bot$, and overall $\res^\Speer=\invec$.
\end{proof}

Denote by $\sigma^\Speer_p$ the local value of $\sigma_p$ for peer $\Speer$ at the beginning of phase $p$.
Denote by $\res^\Speer_p[i]$ the local value of $\res[i]$ for peer $\Speer$ after stage 1 of phase $p$.

\begin{claim} \label{claim:matching_assignment}
For every phase $p$, two \nonfaulty peers $\Speer,  \Speer'$, and bit $i$, one of the following holds.
\begin{enumerate}[$(1_p)$]
    \item $\sigma^\Speer_p(i) =\sigma^{\Speer'}_p(i)$, i.e., both $\Speer$ and $\Speer'$ assign the task of querying $i$ to the same peer, or
    \item $\res^\Speer_p[i] \neq \bot$ or $\res^{\Speer'}_p[i] \neq \bot$. 
\end{enumerate}
\end{claim}
\begin{proof}
By induction on $p$.
For the basis, $p=0$, the claim is trivially true because of the initialization values (specifically, 
property ($1_0$) holds).

For $p\geq 1$. By the induction hypothesis, either $(1_{p-1})$ or $(2_{p-1})$ holds. 
Suppose first that $(2_{p-1})$ holds, i.e., $\res^\Speer_{p-1}[i] \neq \bot$ or $\res^{\Speer'}_{p-1}[i] \neq \bot$. Without loss of generality, assume that $\res^\Speer_{p-1}[i] \neq \bot$. Then, since values are never overwritten, $\res^\Speer_p[i] \neq \bot$, so $(2_p)$ holds as well.

Now suppose that $(1_{p-1})$ holds, i.e., $\sigma^\Speer_{p-1}(i) =\sigma^{\Speer'}_{p-1}(i)$. Let $j$ be an index such that $\sigma^\Speer_{p-1}(i)=\Speer_j$. If both $\Speer$ and $\Speer'$ didn't hear from $\Speer_j$ during phase $p-1$, then both peers will assign the same peer to $i$ in stage 3 of phase $p-1$ (see Line \ref{line:reassign}), so $(1_p)$ holds. If one of the peers heard from $\Speer_j$, w.l.o.g assume $\Speer$ did, then $\res^\Speer_p[i] \neq \bot$. Hence, $(2_p)$ holds.
\end{proof}

Claim \ref{claim:matching_assignment} yields the following corollary.

\begin{corollary} \label{cor:correct_response}
    Every phase $p$\ stage $1$ request received by a \nonfaulty peer is answered with the correct bit values.
\end{corollary}

Next, we show that the protocol never deadlocks, i.e., whenever a \nonfaulty peer waits in stages 2 and 3 (see Lines \ref{line:wait_2} and \ref{line:wait_3}), it will eventually continue.

\begin{claim} \label{claim:one_nonfaulty_termination}
    If one \nonfaulty peer has terminated, then every \nonfaulty peer will eventually terminate.
\end{claim}
\begin{proof}
Let $\Speer$ be a \nonfaulty peer that has terminated. Prior to terminating, $\Speer$ queried all the remaining unknown bits and sent all the bits to every other peer. Since $\Speer$ is \nonfaulty, every other \nonfaulty peer $\Speer'$ will eventually receive the message sent by $\Speer$ and will set $H_p^{\Speer'} = \{1, \dots, k\}$, resulting in $\res^{\Speer'}=\invec$ by Observation \ref{obs:no_missing_mahines_no_missing_bits}. Subsequently, $\Speer'$ will terminate as well.
\end{proof}
\begin{claim} \label{claim:finite_wait}
While no \nonfaulty peer has terminated, a \nonfaulty peer will not wait infinitely for $k-f$ responses.
\end{claim}
\begin{proof}
Let $\mu$ be the least advanced \nonfaulty peer, i.e, for every \nonfaulty peer $\mu'$ either $phase(\mu) < phase (\mu')$, or $phase(\mu) = phase(\mu')$ with $stage(\mu) \leq stage(\mu')$. Note that there are at least $k-f$ \nonfaulty peers. As none of these peers terminate before receiving a request from $\mu$ (premise of the claim), they will send back a response. Hence, $\mu$ will receive responses from at least $k-f$ different peers, and will not wait infinitely.
\end{proof}

The combination of Claims \ref{claim:finite_wait} and \ref{claim:one_nonfaulty_termination} implies that eventually, every \nonfaulty peer satisfies the termination condition (see Line \ref{line:termination_condition}) and subsequently terminates correctly (since it queries all unknown bits beforehand). That is because by Claim \ref{claim:finite_wait} some \nonfaulty peer $\Speer$ will get to phase $\log_{k/f}(n/k)$, or set $H_p^\Speer=\{1, \dots, k\}$ prior to that, and terminate, which will lead to the termination of every \nonfaulty peer by Claim \ref{claim:one_nonfaulty_termination}.


\begin{claim} \label{claim:unkown_bits_per_phase}
At the start of phase $p\geq 0$, every \nonfaulty peer has at most $n\cdot \left(\frac{f}{k}\right)^p$ unknown bits.
\end{claim}
\begin{proof}
    By induction on $p$. Consider \nonfaulty peer $\Speer$.
    For the base step $p=0$ the claim holds trivially by the initialization values.
    
    Now consider $p\geq 1$.
    By the induction hypothesis on $p-1$, $\Speer$ has at most $\hat{n} = n\cdot \left(\frac{f}{k}\right)^{p-1}$ unknown bits at the start of phase $p-1$.
    Since unknown bits are assigned evenly in stage $3$ (see Line \ref{line:reassign}), each peer is assigned $\hat{n}/k$ unknown bits (to be queried during phase $p-1$).
    During stage 2 of phase $p-1$, $\Speer$ waits until $|H_{p-1}^\Speer| \geq k-f$, meaning that $\Speer$ did not receive the assigned bits from at most $f$ peers.
    Hence, at most $\hat{n}/k \cdot f = n\cdot \left(\frac{f}{k}\right)^p$ bits are unknown after stage 2 of phase $p$. The claim follows.
\end{proof}

\def\PROOFA{
By Claim \ref{claim:unkown_bits_per_phase} and since unknown bits are distributed evenly among $\{0, \dots, k-1\}$, every \nonfaulty peer queries at most $\frac{n}{k}\cdot \left(\frac{f}{k}\right)^{p}$ in phase $0\leq p\leq \log_{k/f}(n/k)$ and at most $\frac{n}{k}\cdot \left(\frac{f}{k}\right)^{\log_{k/f}(n/k)} = 1$ additional bits when terminating (By Observation \ref{obs:no_missing_mahines_no_missing_bits}). Hence, the worst case query complexity (per peer) is bounded by 
$$\Query ~\leq~ 1 + \sum_{p=1}^{\log_{k/f}(n/k)} \frac{n}{k}\cdot \left(\frac{f}{k}\right)^{p} ~=~ O\left(\frac{n}{\goodfrac k}\right).$$

We next turn to time analysis.
Consider a peer $\Speer$. For every phase $p$, after $\lceil \frac{n}{k} \cdot (\frac{f}{k})^p / \phi \rceil$ time, every $\phase\ p\ \stage\ 1$ response by a \nonfaulty peer is heard by $\Speer$ (even slow ones),
and stage 2 starts. After that, it takes at most $\lceil n \cdot (\frac{f}{k})^{p+1} /\phi \rceil$ time units for every $\phase\ p\ \stage\ 2$ response to be heard by $\Speer$, allowing it to move to stage 3. Hence, it takes at most $\lceil n/k \cdot (\frac{f}{k})^p /\phi \rceil + \lceil n \cdot (\frac{f}{k})^{p+1} /\phi\rceil$ time for phase $p$ to finish once $\Speer$ started it. Finally, upon termination, $\Speer$ sends $\res^\Speer$ which takes $n$ time. 
Let $p'$ be the phase such that $(n/k) \cdot (\frac{f}{k})^{p'} / \phi = 1$ and $p''$ be the phase such that $(nf/k) \cdot (\frac{f}{k})^{p''} /\phi =1$. Hence, $p' = \log_{k/f}\left(\frac{n}{k\phi}\right)$ and $p'' = \log_{k//f}\left(\frac{nf}{k\phi}\right)$.
Overall the time complexity is 
\begin{align*}
    \Time &\le~ \frac{n}{\phi} + \sum_{p=0}^{\log_{k/f}(n/k)} \left(
    \left\lceil\frac{n}{k} \cdot \left(\frac{f}{k}\right)^p \middle/ \phi \right\rceil
    +
    \left\lceil n \cdot \left(\frac{f}{k}\right)^{p+1} \middle/ \phi \right\rceil
    \right)\\
    &=~ \frac{n}{\phi}+
    \sum_{p=0}^{p'} \frac{n}{\phi k} \cdot \left(\frac{f}{k}\right)^p 
    +
    \sum_{p=p'+1}^{\log_{k/f}(n/k)}1
    ~+~
    \sum_{p=0}^{p''} \frac{n}{\phi} \cdot \left(\frac{f}{k}\right)^{p+1}
    + 
    \sum_{p=p''+1}^{\log_{k/f}(n/k)} 1 
    \\
    &\leq~ \frac{n}{\phi}+ 
    \frac{n}{\goodfrac k \phi}
    ~+~
    \log_{k/f}(\phi)
    ~+~
    \frac{nf}{\goodfrac k \phi}
    ~+~
    \log_{k/f}(\phi /f)
    \\
    &=~ \frac{n}{\phi} + O\left(\frac{n\cdot f}{\goodfrac k \phi} + \log_{\frac{k}{f}}(\phi)\right) ~=~ O\left(\left(\frac{\byzfrac}{\goodfrac} +1\right)\cdot \frac{n}{\phi} + \log_{\frac{k}{f}}(\phi) \right).
\end{align*}
The lemma follows.
} 

From the above discussion, we have the following lemma.
\begin{lemma}
\label{lem: correctness + complexity}
Algorithm~\ref{alg:f_crash_async_down} solves \download\ in the asynchronous setting with at most $f$ crash faults after $\log_{k/f}(n/k)$ phases with $\Query = O(\frac{n}{\goodfrac k})$ and $\Time = O\left((\frac{\byzfrac}{\goodfrac} +1 v)\cdot \frac{n}{\phi} + \log_{\frac{k}{f}}(\phi) \right)$
\end{lemma}
\begin{proof}
    \PROOFA    
\end{proof}



Finally, a slight modification of the protocol yields an improved time complexity, resulting in the following theorem.

\def\PROOFB{
We modify the protocol to get rid of the $\byzfrac/\goodfrac$ factor (which might be high if $\byzfrac$ is not constant). 
Note that the most time consuming component of the protocol is the $\phase\ p\ \stage\ 2$ responses. These responses could be ``short'', i.e., of constant length (in case they are of the form ``me neither'') but they could also be ``long'' (in case they consist of the bits learned by the sender).
To lower the time complexity, we identify a necessary condition for a peer to send its known bits in a $\phase\ p \ \stage\ 2$ response (rather than 
\metoo). We observe that after one time unit, every message is delivered (even by slow peers). Hence, after at most $\lceil \frac{n}{k} \cdot (\frac{f}{k})^p / \phi \rceil$ time units, every $\phase\ p\ \stage\ 1$ response that was sent is delivered (including slow ones).
Therefore, while waiting for $k-f$ $\phase\ p\ \stage\ 2$ responses, it might be the case that a slow $\phase\ p\ \stage\ 1$ response arrives from peer $\Speer'$ eliminating the need for $\phase\ p\ \stage\ 2$ responses regarding peer $\Speer'$. 
The modification needed for this argument to work is that if $\phase\ p\ \stage\ 2$ responses regarding peer $\Speer'$ are no longer necessary because of its $\phase\ p\ \stage\ 1$ response arriving, $\Speer$ is not blocked from continuing. Note that it is easy to see that this modification doesn't affect the correctness of the protocol.
Hence, the only time when a peer must wait for a long $\phase\ p\ \stage\ 2$ response is when the corresponding peer for which a $\phase\ p\ \stage\ 1$ response was not received has crashed (and therefore its $\phase\ p\ \stage\ 1$ response will never arrive). Also note that once a peer crashes in phase $p$, it will not be heard from by any peer in the following phases. resulting in 
\metoo\
responses. Therefore, every peer waits for long $\phase\ p\ \stage\ 2$ responses at most $f$ times.
This results in the complexity being
\[\Time ~\leq~ \frac{n + \beta n}{\phi} + \sum_{p=0}^{\log_{k/f}(n/k)}  \left\lceil\frac{n}{k} \cdot \left(\frac{f}{k}\right)^p \middle/ \phi \right\rceil ~=~ O\left(\frac{n}{\phi} + \frac{n}{\goodfrac k \cdot \phi} + \log_{k/f}(\phi)\right) ~=~ O\left(\frac{n}{\phi} + \log_{k/f}(\phi)\right) .
\]
The theorem follows.
}


\begin{theorem} 
\label{thm:async_multi_crash}
There is a deterministic protocol for solving \download\ in the asynchronous setting with at most $f$ crash faults (for any $f < k$)
with $\Query = O(\frac{n}{\goodfrac k})$, $\Time = O\left(\frac{n}{\phi} + \log_{k/f}(\phi)\right)$ and $\Message = O(nk^2)$.
\end{theorem}
\begin{proof}
\PROOFB    
\end{proof}

 \section{Download in the Asynchronous Model with Byzantine Faults}
\label{sec:async_byz_faults}


In this section, we consider the asynchronous model with \emph{Byzantine} faults, rather than crashes. In this setting, the \download\ problem is still solvable by the naive protocol where each honest machine queries all bits, but it is unclear whether one can do better.
It turns out that this depends on whether $\byzfrac\ge 1/2$ or not. The next two subsections handle these cases.

\subsection{Majority Byzantine Failures (\texorpdfstring{$\byzfrac \ge 1/2$}{more than half})}

When $\byzfrac \geq 1/2$, any asynchronous $\download$ protocol that is resilient to Byzantine faults requires $\Query = \Omega(n)$. Moreover, 
any \emph{deterministic} asynchronous $\download$ protocol resilient to Byzantine faults requires $\Query = n$, namely, the \emph{only} such protocol is the naive one.

\begin{theorem}
\label{thm:async_byz_det_lower_bound}
When $\byzfrac\geq 1/2$, every deterministic asynchronous $\download$ protocol that is resilient to Byzantine faults has $\Query = n$.
\end{theorem}

\commentstart

\commentend

We defer the proof to the appendix \ref{apdx:lower_bound} since we next establish a similar result for randomized protocols (with a slightly weaker bound of $n/2$ instead of $n$). 

One subtle point that makes the randomized lower bound more complicated has to do with the limitations imposed on the adversary due to the fact that honest peers are aware of the minimal number of honest peers in every execution, and are therefore entitled to wait for these many messages. This point deserves further scrutiny. In the asynchronous model, each peer operates in an \emph{event-driven} mode. This means that its typical cycle consists of (a) performing some local computation, (b) sending some messages, and then (c) entering a waiting period, until ``something happens.'' Typically, this ``something'' is the arrival of
a new message. However, the algorithm may instruct the peer $\Speer$ to continue waiting until it receives new messages from $k-f-1$ distinct peers. Since it is guaranteed that at least this many honest peers exist in the execution, this instruction is legitimate (in the sense that it cannot cause the peer to deadlock). Moreover, in case $\Speer$ has already identified and ``blacklisted'' a set $F'$ of peers as failed peers, it is entitled to wait until it receives new messages from $k-f-1$ distinct peers in $V\setminus F'$. This observation restricts the ability of the adversary to delay messages sent by honest peers indefinitely. In particular, at some point during the execution, it may happen that
all honest peers are waiting for new messages from other honest peers
and will not take any additional actions until they do. In such a situation, sometimes referred to in the literature as reaching \emph{quiescence}, the adversary may not continue delaying messages
indefinitely and is compelled to ``release'' some of the delayed messages and let them reach their destination. Throughout, we assume that the adversary abides by this restriction.

\begin{theorem}
For any asynchronous \download\ protocol where $\byzfrac\geq 1/2$, there are executions in which some peer queries more than $n/2$ bits.
\end{theorem}

\begin{proof}
Assume towards contradiction that there exists a randomized protocol $\Protocol$ such that in every execution of $\Protocol$ every peer queries at most $n/2$ bits. Consider the following (types of) executions.
\begin{description}
\item[\textbf{Execution $A$.}] The input is all 0's.
The adversary corrupts the peers $\Speer_2, \dots, \Speer_{k/2}$ and delays the peers  $\Speer_{k/2+1}, \dots, \Speer_{k}$ until $\Speer_1$ terminates. (If 
quiescence is reached before $\Speer_1$ terminates, then the adversary forwards all delayed messages and abandons its attempt to fail the protocol). 
The adversary sets the random string used by the peers $\Speer_2, \dots, \Speer_{k/2}$ to $\hat{r}$ ($\hat{r}$ should be picked in a way that makes the probability of quiescence negligible; we explain later how this is done). The corrupted peers act as they would in an honest execution (except they use $\hat{r}$ set by the adversary). 
\item[\textbf{Execution $B$.}] The input is all 0's, except for one index $i$ ($i$ should be picked according to the random distribution used by $\Speer_1$ such that the probability that $i$ gets queried by $\Speer_1$ is less than $1/2$; we explain later how this is done). The adversary corrupts the peers $\Speer_2, \dots, \Speer_{k/2}$ and delays the peers  $\Speer_{k/2+1}, \dots, \Speer_{k}$ until $\Speer_1$ terminates (same as in execution $A$). The adversary sets the random string used by the peers $\Speer_2, \dots, \Speer_{k/2}$ to the same $\hat{r}$ as in execution $A$. The corrupted peers act as if they are in execution $A$. 
\end{description}
    
Denote by $A_{r,\hat{r}}$ (respectively $B_{r,\hat{r}}$) an execution of type $A$ (respectively $B$) where $\Speer_1$ uses $r$ as its random string and $\Speer_2, \dots, \Speer_{k/2}$ use $\hat{r}$. We denote by $A_r$ (respectively, $B_r$) the execution $A_{r,\hat{r}}$ (resp., $B_{r,\hat{r}}$) where $\hat{r}$ is chosen by the adversary.
    
Note that from the point of view of $\Speer_1$, executions $A_r$ and $B_r$ are indistinguishable if $\Speer_1$ does not query bit $i$. Also note that the adversary's strategy is only valid if $\Speer_1$ does not reach a \emph{quiescent} state, namely, one in which it will not terminate before receiving a message from at least one of the peers $\Speer_{k/2+1}, \dots, \Speer_k$. 
We now show that $\hat{r}$ and $i$ can be chosen such that the probability (over the possible random choices $r$ of $\Speer_1$) of reaching quiescence is negligible and the fraction of pairs $(A_r,B_r)$ in which $\Speer_1$ queries bit $i$ is at most $1/2$.

W.l.o.g, we assume that $\Speer_1$ picks exactly $n/2$ bits to query.
First, for every set $S$ of $n/2$ bits, the adversary knows the probability $p_S$ that $\Speer_1$ will query $S$. For every $i$, denote the probability that bit $i$ gets queried by $\Speer_1$ by $x_i = \sum_{S:i\in S} p_S$.
Note that $\sum_{i=1}^n x_i =\sum_S \frac{n}{2} \cdot p_S = n/2$ (since every set $S$ contributes its probability $p_S$ to the sum
$n/2$ times ). The adversary picks bit $i$ with probability $p_i = \frac{1-x_i}{n/2}$ (note that $\sum_i p_i=1$).
    
Let $\hat{S}$ and $I$ be random variables indicating the random selection of a set $S$ by $\Speer_1$ and an index $i$ by the adversary, respectively. Then
\begin{eqnarray*}
P[I\in \hat{S}] &=& \sum_{i=1}^n P[i\in \hat{S} \wedge I=i] ~=~ \sum_{i=1}^n P[i\in \hat{S}] \cdot P[I=i] ~=~ \sum_{i=1}^n x_i \cdot \frac{1-x_i}{n/2} 
\\
&=& \frac{2}{n} \sum_{i=1}^n x_i \cdot (1-x_i) ~\leq~ \frac{2}{n} \cdot \frac{n}{4} ~=~ \frac{1}{2}~,
\end{eqnarray*}
where the second equality follows since $\hat{S}$ and $I$ are independent and the inequality is derived by the Cauchy–Schwarz inequality.

    
Next, we show that there is a choice of $\hat{r}$ for which the probability (over the possible random choices $r$ of $\Speer_1$) of reaching quiescence in $A_r$ is negligible.
Denote by $Q_{\hat{r}}$ the event of reaching quiescence in $A_{r,\hat{r}}$ (where $\Speer_1$ uses the random string $r$). Denote by $Q$ the event of reaching quiescence in $A_{r,\hat{r}}$ (with the random strings $r$, $\hat{r}$).
Assume towards contradiction that for every value of $\hat{r}$, the probability of $Q_{\hat{r}}$, over the possible random choices $r$ of $\Speer_1$, is $P[Q_{\hat{r}}] \geq 1/n$. Consider an execution $A'_{r,\hat{r}}$, where the input is all 0's, the adversary crashes peers $\Speer_{k/2+1},\dots, \Speer_k$, and the random strings used by the peers $\Speer_1$ and $\Speer_{2},\dots, \Speer_{k/2}$ are $r$ and $\hat{r}$, respectively. Note that if quiescence is reached in $A_{r,\hat{r}}$, then quiescence is also reached in $A'_{r,\hat{r}}$.
Hence, $A'_{r,\hat{r}}$ does not terminate. Noting that
$$P[Q] ~=~ \sum_{\hat{r}} P[Q_{\hat{r}}] \cdot P[\hat{r}] ~\geq~ \frac{1}{n}\sum_{\hat{r}}  P[\hat{r}] ~=~ \frac{1}{n},$$ 
we get that the probability that the algorithm fails (does not terminate) in $A'_{r,\hat{r}}$ is at least $1/n$, a contradiction.
It follows that there must be a choice of $\hat{r}$ such that the probability of $Q_{\hat{r}}$ over the possible random choices $r$ of $\Speer_1$ satisfies $P[Q_{\hat{r}}] < 1/n$. The theorem follows.
\end{proof}

\subsection{Minority Byzantine Failures (\texorpdfstring{$\byzfrac < 1/2$}{less than half})}

We now turn to the case where $\byzfrac < 1/2$.
We first show that in this case, one can do slightly better than the naive protocol even deterministically, presenting a deterministic asynchronous protocol that solves \download\ with query complexity $Q=O\left(\byzfrac  n \right)$.
Next, we show that using randomization yields an improved query complexity.

\subsection{Deterministic \download\ protocol}

We now present a deterministic asynchronous \download\ protocol for Byzantine faults where $\byzfrac< 1/2$.
%
Consider the deterministic synchronous protocol presented in \cite{ABMPRT24}, where a committee $C_i$ of size $2f+1$ is formed for every $i \in [1, n]$, and every \peer $\Speer\in C_i$ queries the bit $\invec[i]$ and broadcasts 
the message $(\Speer,\invec[i]=b_i)$ to all other peers.
In order to adapt that protocol to the asynchronous model, we modify the final part of the protocol, requiring each non-committee peer $\Speer\notin C_i$ to wait until it gets messages $(\Speer',\invec[i]=b)$ with identical value $b$ from at least $f+1$ \peers $\Speer'$, and then decide $\res^\Speer[i] \gets b$.
\begin{theorem}
When $\byzfrac<1/2$, there exists a deterministic  asynchronous protocol that solves \download\ with $Q=O(\byzfrac n)$, $\Time=O\left(\frac{\byzfrac n}{\phi}\right)$ and $\Message = O(f\cdot  n)$.
\end{theorem}
\begin{proof}
To establish correctness, we first note that if a peer $\Speer \notin C_i$ receives $f+1$ messages with identical value $\invec[i]=b$, 
then, since at least one of the messages came from an honest peer, $b=b_i$.
Next, we argue that eventually, every peer $\Speer \notin C_i$ receives $f+1$ messages with identical values for $\invec[i]$. 
This is because $|C_i|=2f+1$, so at least $f+1$ peers are honest. Since honest peers in $C_i$ can only be delayed (and not forced to send an incorrect value), eventually, every non-committee peer $\Speer\notin C_i$ will receive the messages from the ($f+1$ or more) honest peers, which will all be for $b_i$.

Turning to complexity, note that a peer queries once for every committee it belongs to. Since every committee is of size $2f+1$, and committees are selected in round-robin order, every peer belongs to at most $(2f+1)n/k$ committees. Hence, the query complexity is $\Query=O(\byzfrac n)$.
The number of committees that can operate in parallel is $\frac{k}{2f+1} \cdot \phi$, since $k/(2f+1)$ committees will be pairwise disjoint, and every peer can participate in $\phi$ committees at the same time. Hence, the time complexity will be $\Time = \frac{n}{\frac{k\phi}{2f+1}} = O\left(\frac{\byzfrac n}{\phi}\right)$. Every committee sends $2f+1$ messages, so overall $\Message = O(f\cdot  n)$.
\end{proof}

\subsection{Randomized \download\ protocols}
In this section, we show how to adapt the randomized synchronous protocols 3 and 4 of the companion paper \cite{ACKKMP25a} to the asynchronous setting.
We start by defining some notations and data structures that are used in the protocols. Then, proceed to describe the modified protocols. Here, we assume that the size of a messages is $\phi = O(n)$.
\subsubsection{\texorpdfstring{$t$-frequent strings and decision trees}{frequent strings and decision tree}}
Throughout, we refer to the value of a contiguous segment of input bits as a bit string. In our protocols, the input vector $\invec$ is partitioned into $\Inum= \lceil n/\varphi\rceil$ segments of roughly equal length $\varphi$. Hereafter, the $\ell$th segment is denoted by $\invec[\ell,\varphi]$. Peers send queries for a specific segment $\invec[\ell,\varphi]$ of input bits and receive the correct bit string as a response. They then send their findings to all other peers, in a message of the form $\langle \indID,s \rangle$. Two received strings $\langle \indID,s \rangle$ and $\langle \indID',s' \rangle$ are said to be \emph{overlapping} if they correspond to the same segment, namely, $\ell=\ell'$, and \emph{consistent} if in addition $s=s'$. We refer to consistent bit strings received from at least $t$ different peers as \emph{$t$-frequent strings}. Denote by $\FS(MS,t)$ the function that takes a multiset $MS$ of overlapping bit strings and a threshold $t$ and returns the set of $t$-frequent strings found in $MS$.

Consider a set $S=S[\ell,\varphi]$
of overlapping bit strings (that encapsulates possible versions for the same 
segment of the input vector) corresponding to the $\ell$th segment. Note that $S$ might not be entirely consistent, due to messages sent by Byzantine peers. We describe the construction of a \emph{decision tree} for the set $S$ as follows. If $|S|=1$ create a single node labeled by the single-bit string in S. Otherwise, arbitrarily pick two (non-consistent) bit strings $s,s'\in S$
and let $i$ be the first \emph{separating} index in which they differ, i.e., $s[i] \neq s[i']$. Create a root node $v$ and label it with the separating index $i$, i.e., $l(v) =i$.  Let $S^b = \{s\in S \mid s[i]=b\}$, recursively construct $T^0, T^1$ from $S^0, S^1$ respectively, and add an edge from $v$ to the roots of $T_0, T_1$. The resulting structure has different (non-consistent) bit strings at its leaves (as labels) and separating indices at its inner nodes such that querying the input array $\invec$ at those indices (that can be found in the inner nodes), resolves conflicts between the different versions found in $S$ and leaves us with a single bit string that matches those queries. As long as one of the leaves contains the correct bit string, this procedure, referred to as \emph{determining} the decision tree, returns the correct version of the bit string.
For a full technical description, see Protocol \ref{alg:decision_tree}.

\begin{algorithm}
    \caption{\ConstDecTree$(S)$
    \\
    input: Set of strings $S$
    \\
    output: Node 
    labeled tree $T$}
    \label{alg:decision_tree}
\begin{algorithmic}[1]
\State Create a root node $v$
\If{{$|S|>1$}}
    \State Set $i \leftarrow $ smallest index of bit on which at least two strings in $S$ differ
    \State $~label(v)\leftarrow i$   
    \State Set $S^b \gets \{s \in S \mid s[i] = b\}$ for $b\in \{0,1\}$
    \State $T^0\gets$ \ConstDecTree $(S^0)$
    \State $T^1\gets$ \ConstDecTree$(S^1)$

    \State Let $T$ be tree rooted at $v$ with $\leftchild(v)\leftarrow T^0$ and $\rightchild(v)\leftarrow T^1$ 
\Else 
    \State $label(v)\gets s$, where $S=\{s\}$ \Comment{$|S|=1$, $v$ is a leaf}
    \State Let $T$ be a tree consisting of the singleton $v$.
\EndIf
\State Return $T$
\\

\Procedure{\Determine}{T}
\State $J =\{j~|~\exists u \in T \;\text{s.t.}\; j=label(u)\}$
\For {all  $j \in J$ } in parallel
\State $query(j)$
\EndFor
\State Let $v$ be the root of $T$
\While{$v$ is not a leaf}
    \State Let $j = label(v)$.
    \If{$b_j =0$}
        \State Set $v\gets \leftchild(v)$
    \Else
        \State Set $v\gets \rightchild(v)$
    \EndIf
\EndWhile
\State Return $label(v)$
\EndProcedure
\end{algorithmic}
\end{algorithm}

%

In the remainder of this section, we present the modified protocols. The idea here is to change the threshold for entering the decision tree such that in any set of $\gamma k - f$ honest peers, every segment is queried by at least $t$ peers from that set. Other than that modification, we also need to make sure that every peer waits for at most $\gamma k $ messages each time it reads.


\subsubsection{2-cycle protocol}
We start with an overview of Protocol \ref{alg:two_rounds}.
The entire input is partitioned into $\Inum = \lceil n/\varphi\rceil$ segments of size $\varphi$. At the start of the execution, each peer $\Speer$ picks a segment $\indID^\Speer$
u.a.r, then it queries $\indID^\Speer$ and broadcasts the resulting bit string to every other peer. Then, every peer utilizes decision trees to correctly determine the correct bit string for every segment $\indID$ by constructing a decision tree from all the overlapping bit strings corresponding to $\indID$ that were received at least $t = \frac{(\goodfrac -\byzfrac) k}{2 \Inum}$ times.

\begin{algorithm}
\caption{2-Round \download\ with $\beta<1/2$; Code for peer $\Speer$}
\label{alg:two_rounds}
\begin{algorithmic}[1]
\If{ $ k \leq 32(c+1) \ln n/\gamma$} query every bit and return
\EndIf
\If{$k < \sqrt{n/(\gamma-\beta)} \ln n$}
    $\varphi \leftarrow \lceil{\frac{32 (c+1)n \ln n}{(\gamma-\beta) k}}\rceil$
\Else{ ~$\varphi \leftarrow \lceil{32(c+1)\sqrt{n/(\gamma-\beta)}}\rceil$}
\EndIf
\State $t\leftarrow (\gamma-\beta) k/(2\Inum)$
\State Randomly select 
$\indID^\Speer \in [1, \Inum]$
\State Set string $s^\Speer \gets query(\invec[\indID^\Speer, \varphi])$
\State Broadcast $\langle \indID^\Speer,s^\Speer \rangle$
\State {\bf Wait} for a  $\langle \indID,s \rangle$ message from at least $\gamma k$ peers.
\For{{$\indID=1$} to $\Inum$} in parallel

    \State Construct the multiset $S_\indID \gets \{s ~\mid~ \langle \indID,s \rangle \hbox{ received}\}$

\State $T_\indID \gets \ConstDecTree(\FS(S_\indID,t))$

\State $s^\indID \gets \Determine(T_\indID)$
\EndFor
\State 

Output $s^1 s^2 \cdots s^{\Inum}$

\end{algorithmic}
\end{algorithm}

\subparagraph*{Correctness.}
Consider an execution of the protocol, and let $R^\Speer$ be the set of honest peers peer $\Speer$ heard from during the execution. Note that since $R^\Speer$ is decided by the adversary before any peer selects which segment to query, the probability of each peer $\Speer' \in R^\Speer$ to query a specific segment $\indID$ is $1/\Inum$ (i.e., unaffected by the adversary's choice). Also note that $|R^\Speer| \geq (\gamma-\beta)k$ (since $\beta k$ might be slow and another $\beta k$ Byzantine). Denote the number of honest peers, from which peer $\Speer$ heard, that pick the segment $\indID$ by $k^\Speer_\indID$.
The protocol succeeds if $k^\Speer_\indID \ge t$ for every segment $\indID$, 
since every decision tree $T_\indID$ will contain a leaf with the correct string for the segment $\indID$ returned correctly by Procedure \Determine.
\begin{claim}
\label{cl:enough readers}
For constant $c\ge 1$, if $t \geq 8(c+2)\ln n$, then $k^\Speer_\indID \ge t$
for every segment $\indID \in [1, \Inum]$,
with probability at least $1-1/n^c$.
\end{claim}
\begin{proof}
Fix $\indID \in [1,\Inum]$.
The expected number of honest peers from which peer $\Speer$ heard that pick the segment $\indID$ is 
$\Exp[k^\Speer_\indID] =|R^\Speer|/\Inum \ge(\gamma-\beta) k/\Inum=2t$. 
Therefore, applying Chernoff bounds, 
$\Pr[k^\Speer_\indID < t] \leq e^{-t/8}$
$\leq 1/n^{c+2}$ for $t \geq 8(c+2)\ln n$. Taking a union bound over all segments and peers ($\Inum +k \leq 2n$ random variables), the probability that $k^\Speer_\indID < t$ for \emph{any} segment $\indID$ and peer $\Speer$ is less than $1/n^c$.
\end{proof}

\subparagraph*{Query Complexity}
Denote by $\FS_\indID$ a set of $t$-frequent overlapping strings for segment $\indID$.
The cost of querying in Step 2 for segment $\indID$ (where $\FS_\indID$ is the set calculated by the peer) is the number of internal nodes of the decision tree, which is $|\FS_\indID| -1$.
Let $x_\indID$ be the number of overlapping strings received for segment $\indID$ in Step 1 (including copies). Then  $|\FS_\indID|\leq x_\indID/t$. Since each peer sends no more than one string overall, 
$\sum_\indID x_\indID=k$.
Hence, the cost of determining all segments is $\sum_\indID |\FS_\indID| \leq  \sum_\indID x_\indID/t \leq k/t$. 

The query cost per peer, $\Query$, is the cost of determining every segment using decision trees plus the initial query cost; hence $\Query\leq k/t + \varphi$.  Since $t=\frac{(\gamma-\beta) k}{2\Inum}$, $ \Query\leq 2\frac{ \lceil{n/\varphi}\rceil} {\gamma-\beta} + \varphi $.
To satisfy the premise of Claim {\ref{cl:enough readers}}, it is required that
$t = \frac{(\gamma-\beta) k}{2\Inum} \geq  8(c+1) \ln n$. To set the value of $\varphi$, we consider the following cases.
\begin{description}
    \item[\textbf{Case 1}] \textbf{($k\geq \sqrt{n/(\gamma-\beta)}\ln n$).} Set  $\varphi=\lceil{32(c+1) \sqrt{n/(\gamma-\beta)}}\rceil$.
Since $t =\frac{(\gamma-\beta) k}{2 \lceil{\frac{n}{\varphi}}\rceil} \geq \frac{(\gamma-\beta) k}{4 n/\varphi}\geq 8 (c+1) \ln n$, then $t$ satisfies the premise of Claim \ref{cl:enough readers} and
$\Query \leq k/t + \varphi = \frac{k}{\frac{ (\gamma-\beta) k}{2\lceil{\frac{n}{\varphi}}\rceil}} + \varphi = \frac{2 \lceil{n/\varphi\rceil}}{\gamma-\beta} +\varphi=O\left( \sqrt{\frac{n}{\gamma-\beta}}\right)$.
    \item[\textbf{Case 2}] \textbf{($32(c+1) \ln n /\gamma <k < \sqrt{n/(\gamma-\beta)}\ln n $).} Set $\varphi= \lceil{\frac{32 (c+1)n \ln n}{(\gamma-\beta) k}}\rceil$.
    Then $t =\frac{(\gamma-\beta) k}{2 \lceil{\frac{n}{\varphi}}\rceil} \geq \frac{(\gamma-\beta) k}{4 n/\varphi} \geq 8 (c+1) \ln n$, satisfying Claim \ref{cl:enough readers}. Also $k/t \leq  \sqrt{n/(\gamma-\beta)} \ln n/t= O\left( \sqrt{\frac{n}{\gamma-\beta}}\right)$. Thus,
$\Query \leq k/t + \varphi = O\left( \sqrt{\frac{n}{\gamma-\beta}}+ \frac{n \ln n}{(\gamma-\beta) k} \right)$. 
    \item[\textbf{Case 3}] \textbf{($k\leq 32(c+1)\ln n/ \gamma$).} Each peer queries all $n$ bits, resulting in $\Query=n$.
\end{description}
Since the protocol can check to see which case it is in (by inspecting $k$, $n$, and $\gamma$), the overall query complexity is $\Query = O\left(\min \left\{\sqrt{\frac{n}{\gamma-\beta}}+ \frac{n \ln n}{(\gamma-\beta) k}, n\right\}\right)$.
\begin{theorem}
    There is a 2-cycle asynchronous randomized protocol for \download\ with $\Query=O\left( \sqrt{\frac{n}{\gamma-\beta}}+ \frac{n \ln n}{(\gamma-\beta) k} \right)$
    .
\end{theorem}
\subsubsection{Multi-cycle protocol}
In this subsection, we extend Protocol \ref{alg:two_rounds} into a multi-cycle protocol that improves the query complexity in expectation. The first cycle is the same as in Protocol \ref{alg:two_rounds}. In every subsequent cycle $i>1$, the input is partition into $\Inum_i = n/(2^i\varphi)$ segments of size $\varphi_i =2^i\cdot \varphi$. We refer to segments in cycle $i$ as $i$-segments. Note that each $i$-segment is a concatenation of two $(i-1)$-segments. Each peer $\Speer$ picks an $i$-segment $\indID_i^\Speer$ u.a.r and determines its value by constructing decision trees for the two $(i-1)$-segments that compose it from the bit strings (corresponding to each $(i-1)$-segment respectively) received at least $t_{i-1}$ times in cycle $i-1$. Then, it broadcasts the resulting bit string (of $\indID_i^\Speer$) to every other peer. Since the size of the picked segment is doubled each cycle, after $O(\log \Inum)$ cycles, each peer picks the entire input, learns it, and outputs correctly.

\begin{algorithm} 
\caption{Download Protocol with $\tilde{O}(n/\gamma{k})$  Expected Queries and ${O(\log \gamma k)}$ Time $\beta<1/2$}
\label{alg:time-query}
\begin{algorithmic}[1]

\For{$i=$ 0 to $\lg\Inum$}
    \State Randomly pick a segment $\indID \in [1, \Inum_i]$ of size $\varphi_i$
    \If{$i=0$}
        \State Set string $s=query(\invec[\indID, \varphi])$ and broadcast $\langle \indID,s,0\rangle$
    \Else
        \State Wait for $\langle \indID,s,i-1\rangle$ messages from at least $\gamma k$ peers.
        \State $\indID_L \gets 2\indID-1$;$\indID_R \gets 2\indID$
        \For{$u \in \{\indID_L, \indID_R\}$} in parallel 
            \State Construct the multiset $S(u) \leftarrow \{s~\mid~  \text{received a message of the form }\langle u,s,i-1\rangle\}$ 
            \State $t_{i-1}=2^{i-2} \varphi (\gamma k/n)$
            \State $T_u \gets \ConstDecTree(\FS(S(u),t_{i-1}))$ \label{line:construct}
            \State $s_{u} \gets \Determine(T_u)$
        \EndFor
        \State Set $s_{\indID}$ to $s_{\indID_L}s_{\indID_R}$ and broadcast $\langle \indID,s_{\indID},i\rangle$
        
    \EndIf
\EndFor
\State return the determined string for segment $[1,\dots, n]$
\end{algorithmic}
\end{algorithm}

We show correctness by proving the following two key lemmas.

\begin{lemma}\label{lem:high_prob_t_picked}
For every cycle $i \in [0, \log n]$, every segment in cycle $i$ (those of size $\varphi_i$) is picked by at least $t_i$ $\text{honest}$ {peers} w.h.p
\end{lemma}
\begin{proof}
Fix cycle $i$. The number of segments in cycle $i$ is
$\Inum_i$. let $R_i^\Speer$ be the set of honest peers peer $\Speer$ heard from during cycle $i$. Since $R_i^\Speer$ is decided by the adversary before any query is made during cycle $i$,
The expected number of honest peers that pick a given segment at cycle $i$ is 
$E_i=|R_i^\Speer| /\Inum_i= (\goodfrac -\byzfrac)k\cdot\frac{2^i\varphi}{n} \geq  2^i (8(c+2) \ln n)$.
Setting $t_i=E_i/2$, the probability of failure for any one segment, as given by Chernoff bounds (see previous subsection) is no more than $e^{-E_i/8}\leq n^{-c+2}$.

Taking a union bound over the $\sum_{i=0}^{\lg\Inum} \Inum_i < \sum_{i=0}^{\infty} \lceil\frac{n}{\varphi}\rceil \frac{1}{2^i} < n$ segments over all cycles $i$ and $k\leq n$ peers, 
the probability of any failure in any cycle is less than $n^2 \cdot n^{-(c+2)}\leq n^{-c}$ for $c \geq 1$. 
\end{proof}
Denote by $\indID_i^\Speer$ the segment ID picked by peer $\Speer$ at cycle $i$.
\begin{lemma} \label{lem:fast_correct_value}
In every cycle $i\in [0,\lg\Inum]$, every honest peer $\Speer$ learns the correct value of $\invec[\indID_i^\Speer,\varphi_i]$ w.h.p.
\end{lemma}
\begin{proof}
By induction on the cycles.
The base case, cycle $u=0$, is trivial since every honest peer that picks a segment queries it completely. 
For the inductive step, cycle $i>0$, let $\indID = \indID_i^\Speer$ be the segment  picked by peer $\Speer$ in cycle $i$. During cycle $i$, $\Speer$ splits $\indID$ into two subsegments $\indID_L, \indID_R$ of size $\varphi_{i-1}$, and construct the $t_{i-1}$ frequent sets, $\FS_u =\FS(S(\indID_u),t_{i-1})$ for $u\in\{L,R\}$.
    By Lemma {\ref{lem:high_prob_t_picked}}, the segments $\indID_L, \indID_R$ were each picked by a least $t_{i-1}$ honest peers w.h.p during cycle $i-1$. By the inductive hypothesis, those peers learned the correct string of segments $\indID_L, \indID_R$, $s_L, s_R$ respectively, and then broadcast $s_L, s_R$ to all other peers. Thus in cycle $i$, peer $\Speer$ see both $s_L$ and $s_R$ at least $t_{i-1}$ times w.h.p, so $s_L \in \FS_L$ and $s_R \in \FS_R$. Hence, the decision trees built for $\FS_L$ and $\FS_R$ will return $s_L$ and $s_R$ respectively, and peer $\Speer$ will learn the correct value for segment $\indID$, $s_{\indID}=s_Ls_R$.
\end{proof}

The correctness of the protocol follows from observing that at the last cycle $\indID=\lg \Inum$, the segments are of size $\varphi_\indID = 2^\indID\varphi = n$. Hence, by Lemma {\ref{lem:fast_correct_value}}, every peer learns the entire input.
\subparagraph*{Query complexity.}
%
Let $\indID$ be an $(i+1)$-segment that is a concatenation of two $i$-segments $\indID_L$ and $\indID_R$. We denote by $num_{i+1}(\indID)$ the total number of  strings received for the subsegments $\indID_L$ and $\indID_R$ in cycle $i$, i.e., $num_{i+1}(\indID)=x_L +x_R$, where $x_u$ for $u \in \{L,R\}$ denotes the number of strings received for the subsegment $\indID_u$ during cycle $i$.
Let $m_x$ be the number of segments $\indID$ in cycle $i+1$ for which $num_{i+1}(\indID) = x$.
Formally, 
\begin{equation*}
    m_x = |\{\indID \mid \text{received exactly $x$ messages of the form }\langle\indID_L, s, i\rangle \text{ or } \langle\indID_R,s,i\rangle \}|
\end{equation*}
The probability of picking a segment $\indID$ in cycle $i+1$ with $num_{i+1}(\indID) = x$ strings is $m_x$ divided by the total number of segments, or 
$m_x/\Inum_{i+1}$. The expected cost of querying in Step $i+1$ is therefore
\begin{equation*}
\sum_x\frac{m_x}{\Inum_{i+1}} \cdot \frac{x}{t_i}
~=~ \sum_x \frac{m_x}{\Inum_{i+1}} \cdot \frac{2x}{(\gamma-\byzfrac) k/\Inum_i}
~=~ \frac{2\Inum_i}{(\gamma-\byzfrac) k \cdot \Inum_{i+1}} \sum_x m_x\cdot x
~\leq~ \frac{4}{\gamma-\byzfrac}~,
\end{equation*}
where the last inequality follows since each peer broadcasts at most one string, so $\sum_x m_x \cdot x\leq k$.

Step 0 requires $\varphi=O(n\log n/((\gamma-\byzfrac) k))$ queries.
The expected cost of querying is $O(1/(\gamma-\byzfrac))$ per step $i>0$ per peer, and there are fewer than $\log n$ steps, so the total expected query cost is at most $O(n\log n/((\gamma-\byzfrac) k))$. 

\begin{theorem}\label{thm:fast_expected_optimal}
There is a $O\left(\log \left(\frac{\goodfrac k}{\ln n}\right)\right)$-cycle protocol which w.h.p. computes \download\ in the point-to-point model with expected query complexity $\Query=O(n \log n/(\gamma - \byzfrac k))$ and message size $O(n)$. 
\end{theorem}




\section{Application: Efficient Blockchain Oracles}
\label{sec:application}
Blockchain systems \cite{nakamoto2009bitcoin} have seen a rise in popularity due to their ability to provide both transparency and strong cryptographic guarantees of agreement on the order of transactions, without the need for trusted third-party entities. More general computational abilities have also been well sought out for blockchains. Smart contracts \cite{Szabo1997FormalizingAS} fulfill that need by providing users of the blockchain a way to run programs \emph{on} the blockchain that ensures reliable and deterministic execution while providing transparency and immutability of both the code of the program and its state(s). Note that since the execution is required to be deterministic, i.e., every node must produce the same result, smart contracts are restricted to accessing \emph{on-chain} data (that has been agreed upon), as \emph{off-chain} data may introduce non-determinism to the execution. 

\emph{Blockchain oracles} \cite{Astraea:Oracle, chainlink, supra2024blockchain} are components of blockchain systems that provide multiple services that support and extend the functionality of smart contracts (and other on-chain entities). The most important and fundamental service a blockchain oracle provides is bridging between the on-chain network and off-chain resources \cite{ocr, DORA}, providing smart contracts access to external data without introducing non-determinism into their execution. We focus on this service and artificially consider it to be the sole responsibility of a blockchain oracle. 
%
In the remainder of this section, we explain in detail a possible application of the DR model and the \download\ problem for improved query efficiency within the context of \emph{blockchain oracles}.

\subparagraph*{Blockchain oracles general structure.} 
Blockchain oracles consist of an on-chain component and an off-chain component. The off-chain component encompasses the different data sources that store the required external information (e.g., stock prices, weather predictions) and the network of nodes in charge of retrieving that information and transmitting it to the on-chain component.
The on-chain component can be thought to be (but is not necessarily) a smart contract that is responsible for verifying the validity of the report, making the information public on the blockchain it is hosted on, and using it for its execution.

Formally, the off-chain component consists of two parts: an asynchronous\footnote{The network is sometimes assumed to be \emph{partially synchronous} in blockchain oracles.}
\emph{oracle network}, with peers (nodes) $v_i$, $i\in [1,k]$, 
capable of exchanging direct messages among themselves,
and \emph{data sources} $DS_j$, $j\in [1,m]$, each storing an array $\invec_j$ of $n$ variables in which the on-chain component is interested.
Each peer can read the $i$-th cell from the $j$-th data source $DS_j$ by invoking $\DataQuery(i,j)$. 
A fraction of up to $\byzfrac_t \leq 1/3$ of the 
peers may be Byzantine, and a fraction of up to $\byzfrac_d\leq 1/2$ of the $m$ data sources may be Byzantine.
Denote the on-chain component by $SC$. 

The goal of blockchain oracles, as mentioned above, is to pull information from external sources and push a final value on-chain. There are a few difficulties that may arise when trying to develop such a system. First, 
it might be the case that different data sources report slightly different values (e.g., prices of a specific stock), even if all of them act honestly. Moreover, corrupted data sources might provide false and even inconsistent values (providing some nodes with value $\alpha$ and another with $\alpha'$). The system needs to pick a final value in a way that (1) represents the range of honest values pulled from honest data sources and (2) malicious players (both data sources and peers) cannot force the system to pick a final value that does not represent the range of honest values.

\subparagraph*{The Oracle Data Delivery (ODD) problem.} 
Denote by $\cH_{ds}$ the set of honest data sources. Let 
$h_{\min}(i) =\min_{j\in \cH}\{\invec_j[i]\} $ and 
$h_{\max}(i) = \max_{j\in \cH}\{\invec_j[i]\}$. 
The \emph{honest range} of $i$ is the range $\sigma(i)=[h_{\min}(i), h_{\max}(i)]$.
The \emph{ODD problem}
requires the on-chain $SC$ to publish an array of values $\res$ to the target blockchain such that
$\res[i] \in \sigma(i)$, for every $i\in [n]$.

A blockchain oracle protocol can generally be split into three distinct steps: (1) collecting data, (2) reaching an agreement on the collected data, and (3) deriving and publishing a final value.
Note that this abstraction is the minimum required abstraction to capture the operation of blockchain oracle protocols such as OCR \cite{ocr} and DORA \cite{DORA}\footnote{These protocols have many additional technical aspects, different structures, and different ways of handling steps (2) and (3). As our focus is on improving step (1), we may w.l.o.g. assume the abstract structure.}.

We now show how our $\download$ protocols can be used to significantly reduce the cost of the \emph{Oracle Data Collection (ODC)}, i.e, step (1) of blockchain oracles. 

\subparagraph*{Improving ODC by blockchain oracles via $\download$}

Current protocols perform the 
data collection
step by the following ODC process: 

\noindent
For every node:
\begin{itemize}
\item Pick $2m\byzfrac_d +1$ data sources into a set $ADS$.
\item Perform $o_{i,j}\gets \DataQuery(i,j)$, for every $i\in[1,n]$ and $j\in ADS$.
\item Calculate the median $o_i\gets median(\{o_{i,j} \mid j\in ADS\})$ and proceed to step (2).
\end{itemize}

The results of \cite{ocr,DORA}, cast in our abstract formulation, yield the following result.

\begin{theorem}
{\bf \cite{ocr,DORA}}
The ODC process guarantees that $o_i \in \sigma(i)$ for every $i\in [1,n]$ and has total query cost $O(m n k)$
and worst case individual query cost $\Query=O(mn)$.
\end{theorem}

Instead, we propose utilizing the guarantees of $\download$ protocols, namely, that for an honest data source $DS_j$,
the output of each peer is exactly $\invec_j$, to construct the following modification of the ODC steps.

\begin{itemize}
\item For every node, pick $2m\byzfrac_d +1$ data sources into a set $ADS$.
\item For every data source $j\in ADS$, run a $\download$ protocol (denote the result for cell $i$ from data source $j$ by $o_{i,j}$).
\item Calculate the median $o_i\gets median(\{o_{i,j} \mid j\in ADS\})$ and proceed to step (2).
\end{itemize}

It is easy to verify that this modified construction yields the following.

\begin{theorem}
The \download-based ODC process guarantees $o_i \in \sigma(i)$ for every $i\in[1,n]$ and takes $\tilde{O}(m n)$ 
total queries and $\Query=\tilde{O}(mn/k)$ w.h.p.
\end{theorem}

Note that the $\download$ protocol presented in this paper assumes a binary input array, but this can be extended to numbers via a relatively simple extension. 
However, it is important to note
that our solution relies on the following restrictive assumption.
For two honest peers $v, v'$, if both $v$ and $v'$ issue the query $\DataQuery(i,j)$, then they get the same result, for every $i\in [1,n]$ and honest data source $j$ (i.e., the data does not change if queried at different times).
Getting rid of this assumption and solving the problem efficiently for dynamic data is left as an open problem for future study.

\bibliographystyle{plain}
\bibliography{references}
\begin{appendix}
    \section{Deferred Deterministic Lower Bound Proof}
    \label{apdx:lower_bound}
    Here, we provide the proof for Theorem \ref{thm:async_byz_det_lower_bound}.
    \begin{proof}
Let $P$ be a deterministic protocol that solves \download\ in the asynchronous setting with $\byzfrac\geq 1/2$ fraction of Byzantine faults and $\Query\leq n-1$.
Let $f=\beta k$ be the upper bound on the number of Byzantine faults.

We say that $P$ is \emph{$f$-vulnerable} if there exists a \emph{synchronous} execution $\EXEC$ of $P$, a set $F$ of at most $f$ peers that crash in $\EXEC$, and a peer $\Speer$ that is \nonfaulty in $\EXEC$, such that $\Speer$ runs forever in $\EXEC$ (i.e., $\EXEC$ is an infinite execution). The \download\ problem requires each \nonfaulty peer to terminate within a finite time. Therefore, we assume hereafter that $P$ is not $f$-vulnerable. In other words, in every synchronous execution $\EXEC$ of $P$ with up to $f$ crashes, all \nonfaulty peers terminate within finite time.

For a set $F$ of at most $f$ \peers, let $\EXEC_s^F$ be the synchronous execution with input $\invec = [0, \ldots, 0]$, where the \peers in $F$ crash, let $t^F_\Speer$ be the (finite) time in which $\Speer \notin F$ terminates in $\EXEC^F_s$ and let $i^F_\Speer$ be a bit that is not queried by $\Speer$ in $\EXEC^F_s$ (which exists because $\Query<n$). 

Let $R$ $=\{\Speer_{k-f+1},\ldots,\Speer_k\}$ 
and let $\invec'$ be a bit array such that, in the execution where the peers of $R$ are Byzantine (i.e., $F=R$), $\invec'[j] = 0$, for every $j\neq i^R_{\Speer_{1}}$, and $\invec'[i^R_{\Speer_1}] = 1$ (i.e, $\invec$ and $\invec'$ differ only in  $i^R_{\Speer_1}$).

Consider the execution $\EXEC$ where the input is $\invec'$, and the adversary applies the following strategy: 
\begin{enumerate}
\item
It slows down the outgoing messages of \peers 
in $R$ by $t^R_{\Speer_1}+1$ 
time units (sufficient so that 
$\Speer_1$ terminates before any message sent by these peers is received by $\Speer_1$, which exists 
since the protocol is not $f$-vulnerable). 
\item
It corrupts the peers 
$F=\{\Speer_2,\ldots,\Speer_{k-f}\}$, and makes them behave as if the input is $\invec$. Note that this is possible since $k-f < f$.
\end{enumerate}
By definition of $\EXEC$, up until time $t^R_{\Speer_1}$,
the \peer $\Speer_1$ (which is \nonfaulty in both executions):
\begin{enumerate}
\item 
will not receive messages from any \peer $\Speer \in$ $R$ in $\EXEC$, as in $\EXEC_s^R$,
\item 
will receive the same messages from the \peers of $F$ in $\EXEC$ as it receives in $\EXEC_s^R$,
\item 
will receive the same responses from the \Source in $\EXEC$ for every query it makes (to bits other than $i^R_{\Speer_1}$) as in $\EXEC_s^R$.
\end{enumerate}
Hence, in both executions $\EXEC_s^R$ and $\EXEC$, $\Speer_1$ will terminate at time $t^R_{\Speer_1}$ without querying bit $i^R_{\Speer_1}$. Therefore, the executions $\EXEC$ and $\EXEC^R_s$ are indistinguishable from the point of view of $\Speer_1$. Hence, it will output the same output in both executions, in contradiction to the correctness requirement.
\end{proof}
\end{appendix}
 
\end{document}